\documentclass[12pt]{article}
\usepackage{amsmath,amsfonts,amsthm,amssymb,bm,bbold,hyperref}

\newtheorem{thm}{Theorem}[section]
\newtheorem*{thm*}{Theorem}
\newtheorem{lemma}[thm]{Lemma}
\newtheorem{prop}[thm]{Proposition}
\newtheorem*{rmk*}{Remark}

\begin{document}

\title
      {An estimate for the average spectral measure of random band matrices}
\author{Sasha Sodin}


\maketitle

\begin{abstract}
For a class of random band matrices of band width $W$, we prove regularity of the average
spectral measure at scales $\epsilon \geq W^{-0.99}$, and find its asymptotics at these
scales.
\end{abstract}

\section{Introduction}

\footnotetext[1]{\scriptsize School of Mathematics,
Institute for Advanced Study, Einstein Dr., Princeton, NJ 08540, USA. \\
E-mail: sodinale@ias.edu.\\
Supported by NSF under agreement DMS-0635607.}

Define a random operator $H$ on $\ell^2(\mathbb{Z})$ via
\begin{equation}
H(u, v) =
    \begin{cases}
        \frac{\pm 1}{2 \sqrt{2W-1}},     &u < v \leq u+W \\
        H(v, u),   &u > v \geq u-W \\
        0,         &\text{otherwise}~,
    \end{cases}
\end{equation}
where the random signs are independent, and the width $W \in \mathbb{N}$ is a large parameter. The integrated density of states
$N: \mathbb{R} \to [0, 1]$ is defined by
\begin{equation}
N(E_0) = \lim_{\epsilon \to 0} \, \frac{1}{\pi}  \Big\langle \int_{-\infty}^{E_0} \Im (H - E - i \epsilon)^{-1}(0, 0)  \, dE \, \Big\rangle~,
\quad E_0 \in \mathbb{R}~,
\end{equation}
where $\langle \cdot \rangle$ denotes average over the randomness. It is known that
the limit exists for almost every $E_0 \in \mathbb{R}$; it is equal to the distribution function
of the average spectral measure of $H$ (corresponding to the vector $\delta_0$.)

If $N$ is (Radon--)differentiable, its derivative is called the density of states, and is denoted
by $\rho$. In this case,
\begin{equation}
\rho(E_0) = \lim_{\epsilon \to +0} \frac{1}{\pi} \Big\langle \Im (H - E_0 - i \epsilon)^{-1}(0, 0)  \, \Big\rangle~;
\end{equation}
the existence of the density of states is equivalent to the existence of the limit, and
to the boundedness of the expression under the limit. It is believed that the density
of states exists for any $E$, is bounded uniformly in $W$, and admits an asymptotic series
\begin{equation}\label{eq:formal}
\rho(E_0) \sim a_0(E_0) + \frac{a_1(E_0)}{W} + \frac{a_2(E_0)}{W^2} + \cdots~,
\end{equation}
where for example
\begin{equation}\label{eq:a0}
a_0(E_0) = \frac{2}{\pi} \sqrt{1 - E_0^2}~.
\end{equation}
In fact, there is a natural perturbation expansion using the so-called self-energy
renormalisation that yields the terms
of (\ref{eq:formal}) one by one; it is equivalent to the one we describe in Section~\ref{s:meth}
(see Spencer \cite{Sp} for the precise definitions.)

However, even the first terms of the series (\ref{eq:formal}) have not yet been rigorously
justified, and it is still unknown whether the density of states exists and is bounded uniformly
in $W \to \infty$. If the entries of $H$ were replaced with random variables with
absolutely continuous distribution with bounded density, Wegner's estimate would show that
the density of states exists; however, even in this case, Wegner's argument only yields the width-dependent
bound $\rho(E_0) \leq C \sqrt{W}$.

We are interested in a simpler problem, namely, the behaviour of the expression
$\Big\langle (H - E_0 - i \epsilon)^{-1}(0, 0)  \, \Big\rangle$ for small $\epsilon > 0$
(depending on $W$.)

\begin{thm*} For $-1 < E_0 < 1$ and $ \epsilon \geq W^{-0.99} $~,
\begin{equation}\label{eq:1term}
  \left|\Im  \Big\langle (H - E_0 - i \epsilon)^{-1}(0, 0)  \,
    \Big\rangle - \Im \int \frac{a_0(E) dE}{E-E_0-i\epsilon} \right|
    \leq \frac{C(E_0)}{W}~,
\end{equation}
with $a_0$ as in (\ref{eq:a0}), and $C(E_0)>0$ independent of $W$ and bounded on any interval
$(-1+\delta, 1-\delta)$. In particular,
\begin{equation}\label{eq:est}
\left| \Im \Big\langle (H - E_0 - i \epsilon)^{-1}(0, 0)  \, \Big\rangle \right| \leq C(E_0)~.
\end{equation}
\end{thm*}

When $\epsilon > 0$ is fixed and $W \to \infty$, the asymptotics (\ref{eq:1term}) follows
from the results of Bogachev, Molchanov, and Pastur \cite{BMP}. For $\epsilon \gtrsim W^{-1/3}$,
the asymptotics (\ref{eq:1term}) follows from the recent result of Erd\H{o}s and Knowles \cite{EK}
(see also \cite{EK2} for an extension to more general random band matrices.) These results are based
on rigorous control of the perturbation series in $W$. As observed by Erd\H{o}s and Knowles, the
perturbative series they use diverges for $\epsilon \lesssim W^{-2/5}$.

On the other hand, the methods of \cite{BMP} and \cite{EK,EK2} allow to handle the
(more difficult and physically more interesting) quantity
\begin{equation}\label{eq:abs}
\Big\langle |(H - E_0 - i \epsilon)^{-1}(0, 0)|^2  \, \Big\rangle~,
\end{equation}
and \cite{EK,EK2} actually control the quantum dynamics $\exp(itH)$ for $t \lesssim W^{1/3}$.

In \cite{ESchY}, Erd\H{o}s, Yau, and Yin used different methods to prove an analogue of (\ref{eq:1term}) for finite band
matrices of size $N \times N$, for $\epsilon \gtrsim W^{-1} \,\, \log^C N$, and with error term
$\lesssim (W \epsilon)^{-1/2} \log^C N$.  They also control the fluctuations of
\[ (H - E_0 - i \epsilon)^{-1}(0, 0)  \]
at these scales.

In this note we go back to the perturbation series (\ref{eq:formal}), and introduce a regularisation
procedure which allows to justify it at scales $\epsilon \geq W^{-0.99}$. For now, we only deal with the average
spectral measure; however, we believe that the method could also be applicable to (\ref{eq:abs}).
One advantage of the perturbative method is that it allows to get the (optimal) error term $O(1/W)$; actually,
our method can be used to justify the first $W^{0.99}$ terms in the asymptotic expansion, and thus
obtain an approximation with error term which is exponentially small in the width $W$.

We refer to the survey of Spencer \cite{Sp} for other results and problems pertaining to random band matrices.

\vspace{2mm}\noindent
The paper is organised as follows. In Section~\ref{s:meth} we describe the formal perturbation
series. In Section~\ref{s:tech} we collect several technical statements which are used in the proof.
Section~\ref{s:diags} briefly summarises the topological classification of paths, based on \cite{FS}.
In Section~\ref{s:emb} we estimate the contribution of every equivalence class to the perturbation
series. The proof of the theorem appears in Section~\ref{s:pf}, and is followed by the concluding
remarks of Section~\ref{s:rem}.

\vspace{2mm}\noindent
{\bf Acknowledgment.} Tom Spencer encouraged me to study the average spectral measure on
short scales, suggested several crucial steps in the argument, and commented on a preliminary
version of this paper. I thank him very much.

\section{Formal perturbation series, and combinatorial preliminaries}\label{s:meth}

Denote by $\mathbb{Z}(W)$ the graph having $\mathbb{Z}$ as its set of vertices,
and $u \sim v$ if $0 < |u - v| \leq W$.

Let $T_n$ and $U_n$ denote the Chebyshev polynomials of the first and second kind,
respectively:
\begin{equation}
T_n(\cos \theta) = \cos (n \theta)~, \quad U_n(\cos \theta) = \frac{\sin((n+1)\theta)}{\sin \theta}~;
\end{equation}
we formally set $U_{-2} \equiv U_{-1} \equiv 0$. Let
\begin{equation}\begin{split}\label{eq:unw}
U_{n,W} &= U_n - \frac{1}{2W-1} U_{n-2}~.
\end{split}\end{equation}
Then (see \cite[Lemma~2.7]{jmp})
\[ U_{n,W}(H)(u_0, u_n) = (2W-1)^{-n/2} \sum H(u_0, u_1) H(u_1, u_2) \cdots H(u_{n-1}, u_n)~, \]
where the sum is over $(n+1)$-tuples $(u_0, u_1, u_2, \cdots, u_n)$ such that
$u_j \neq u_{j+2}$, $j = 0, 1, 2, \cdots, n-2$ ({\em non-backtracking} paths.) In particular,
\begin{equation}\label{eq:unw_paths}
\langle U_{n,W}(H)(u_0, u_n) \rangle = (2W-1)^{-n/2} \text{Paths}_n(u_0, u_n)~,
\end{equation}
where $\text{Paths}_n(u_0, u_n)$ is the number of paths from $u_0$ to $u_n$ in $\mathbb{Z}(W)$
which pass every edge an even number of times and satisfy the condition above (in particular,
$\text{Paths}_n(u, v) = \delta_{uv}$.) Combining (\ref{eq:unw}) with (\ref{eq:unw_paths}), we
obtain:
\begin{equation}
\langle U_{n}(H)(u_0, u_n) \rangle = (2W-1)^{-n/2} \sum_{m = 0}^{n/2} \text{Paths}_{n-2m}(u_0, u_n)~,
\end{equation}
and hence
\begin{equation}
\langle U_{n}(H)(u_0, u_0) \rangle = (2W-1)^{-n/2} \sum_{m = 0}^{n/2} \text{Paths}_{n-2m}(u_0, u_0)~.
\end{equation}
Observe that $\text{Paths}_{m}(u_0, u_0)$ does not depend on $u_0 \in \mathbb{Z}$, therefore
we denote it simply $\text{Paths}_{m}$. Also, $\text{Paths}_{m} = 0$ if $m$ is odd.

Next, $T_n = (U_n - U_{n-2})/2$, therefore
\begin{equation}\label{eq:reprtn}\begin{split}
&\langle T_{n}(H)(0, 0) \rangle \\
&\quad= \frac{1}{2(2W-1)^{n/2}} \sum_{m = 0}^{n/2}
    \Big\{ \text{Paths}_{n-2m} - (2W-1)\,\text{Paths}_{n-2m-2} \Big\}~,
\end{split}\end{equation}
where we formally set $\text{Paths}_{m}=0$ for $m < 0$.

It will be convenient to rewrite $\Big\{ \text{Paths}_{n-2m} - (2W-1)\,\text{Paths}_{n-2m-2} \Big\}$ in
a different form. Let $\text{Paths}^0_n$ be the number of paths which pass every edge an even number of times
and satisfy the strengthened non-backtracking condition $u_0 \neq u_2$, $u_1 \neq u_3$, \dots, $u_{n-1} \neq u_1$. Then
\[ \text{Paths}_n = \text{Paths}^0_n + (2W-2) \sum_{1 \leq j < n/2} (2W-1)^{j-1} \text{Paths}^0_{n-2j}~. \]
Therefore
\begin{equation}\label{eq:paths_paths0}
\text{Paths}_n - (2W-1) \text{Paths}_{n-2}
    = \text{Paths}^0_n - \text{Paths}^0_{n-2}~.
\end{equation}
To make this identity valid for all $n \geq 1$, we formally set $\text{Paths}^0_0 = 2W-1$,
$\text{Paths}^0_{-1} = 0$.

Now, $\delta(E - E_0)$ admits a formal expansion
\begin{equation}\label{eq:reprkernel}
\delta(E - E_0) \sim \frac{1}{\pi \sqrt{1-E_0^2}} \left\{ 1 + 2 \sum_{n = 1}^\infty T_n(E_0) T_n(E) \right\}
\end{equation}
We shall discuss a regularised version of this series in the next section; for now, we remark that
(\ref{eq:reprkernel}) is a rigorous identity in $L_2(-1, 1)$, since $T_n$ are the orthogonal
polynomials with respect to the measure $\frac{dE}{\pi \sqrt{1-E_0^2}}$.

Convolving (\ref{eq:reprkernel}) with $\rho$, we obtain:
\begin{equation}\begin{split}
\rho(E_0)
    &\sim \big\langle \delta(H - E_0)(0, 0) \big\rangle\\
    &\hspace{-15pt}= \frac{1}{\pi \sqrt{1-E_0^2}}  \left\{ 1 + 2 \sum_{n = 1}^\infty T_n(E_0) \langle T_n(H)(0, 0) \rangle \right\} \\
    &\hspace{-15pt}= \frac{1}{\pi \sqrt{1-E_0^2}}  \left\{ 1 + \sum_{n = 1}^\infty
        \frac{T_{2n}(E_0)}{(2W-1)^{n}} \sum_{m = 0}^{n} \left[ \text{Paths}^0_{2m} - \text{Paths}^0_{2m-2} \right] \right\}~.
\end{split}\end{equation}
Picking the addends which are not divided by powers of $2W-1$, we see that the leading term is
\begin{equation}
\frac{1}{\pi \sqrt{1-E_0^2}} \left\{ 1 + \frac{T_2(E_0)}{2W-1} \times (-2W+1) \right\}
    = \frac{2}{\pi} \sqrt{1 - E_0^2}~.
\end{equation}

In this way one can also obtain the full expansion of the form (\ref{eq:formal}) which, unfortunately, diverges.

One may introduce a regularisation factor and consider the expression
\begin{multline}\label{eq:stieltjes_reg}
\frac{1}{\pi} \langle \Im (H - E_0 - i\epsilon)^{-1}(0, 0) \rangle\\
\sim \frac{1}{\pi \sqrt{1-E_0^2}}
    \left\{ 1 + 2 \sum_{n = 1}^\infty e^{-n \epsilon} T_n(E_0) \langle T_n(H)(0, 0) \rangle \right\}~,
\end{multline}
which corresponds to the density of states averaged over an interval of width $\approx \epsilon$ about $E_0$
(more precisely, the convolution of the density of states with an approximate $\delta$-function of
width $\epsilon$.) However, it also diverges for sufficiently small $\epsilon$. The reason is the large contribution to
$\langle T_n(H) \rangle$ of the part of the spectrum of $H$ outside $[-1, 1]$. Namely, it is known \cite{edge}
that a part of the spectrum lies at distance $\approx W^{-4/5}$ from $[-1, 1]$; the polynomials $T_n$ grow as
\[ |T_n(\pm (1+\delta))| \approx \exp(\sqrt{2\delta} \, n)~, \]
therefore the series above can not converge for $\epsilon \ll W^{-2/5}$.

Here we propose a different regularisation:
\[ \frac{1}{\pi \sqrt{1-E_0^2}}
    \left\{ 1 + 2 \sum_{n = 1}^\infty \varphi(n\epsilon) T_n(E_0) \langle T_n(H)(0, 0) \rangle \right\}~, \]
and justify its convergence for $\epsilon \geq W^{-0.99}$. For reasonable $\phi$, we shall see
(in Lemma~\ref{l:pois}) that it also represents an average of
the density of states over an interval of width $\approx \epsilon$ about $E_0$. Therefore {\em a posteriori}
we obtain an expansion for
\[ \frac{1}{\pi} \langle \Im (H - E_0 - i\epsilon)^{-1}(0, 0)\rangle~.\]

The main step is to pick $\varphi$ carefully. In particular, the argument above shows that $\varphi$ has
to decay faster than exponentially at infinity. As we shall see, it will also be convenient to have $\varphi$
positive-definite and analytic.

\section{Some auxiliary statements}\label{s:tech}

Let $\varphi: \mathbb{R} \to \mathbb{R}_+$ be a smooth even function
such that $\varphi(0) = 1$, $\varphi$ decreases to zero on
$\mathbb{R}_+$, and $\hat{\varphi} \geq 0$. We shall take
\begin{equation}\label{eq:phiq}
\varphi_q(t) = \frac{1}{A_q} \int_{-\infty}^\infty \exp \left[ - s^{2q} - (t-s)^{2q} \right] \, ds~,
\end{equation}
where $q$ is a (large) integer parameter which we shall choose later, but the next lemma
will be applicable in the general setting.

For $\epsilon > 0$ and $E_0 = \cos \theta_0 \in (-1, 1)$, set
\[ f_{E_0, \epsilon}(E) = 1 + 2 \sum_{n\geq1} \varphi(n\epsilon) T_n(E_0) \, T_n(E)~. \]
Dirichlet's criterion shows that the series converges.

\begin{lemma}\label{l:pois}
\begin{multline}\label{eq:l.pois}
f_{\cos \theta_0, \epsilon}(\cos \theta) \\
    = \frac{1}{2\epsilon} \sum_{m=-\infty}^\infty
    \left\{ \hat{\varphi} \left( \frac{1}{\epsilon} \left[ m - \frac{\theta + \theta_0}{2\pi} \right] \right)
    + \hat{\varphi} \left( \frac{1}{\epsilon} \left[ m - \frac{\theta - \theta_0}{2\pi} \right] \right) \right\}~.
\end{multline}
\end{lemma}
Note that, if $\hat\varphi \geq 0$, the right-hand side is an approximate $\delta$-function in $\theta$
of width $\approx \epsilon$.

\begin{proof}
First,
\[ f_{\cos \theta_0, \epsilon}(E) = \frac{1}{2} \sum_{n = -\infty}^\infty g(n)~,\]
where
\[ g(t) = \varphi(t\epsilon)
    \Big\{ \exp(it(\theta + \theta_0)) + \exp(it(\theta-\theta_0)) \Big\}~.\]
Now,
\begin{equation}
\begin{split}
\hat{g}(\xi)
    &= \int_{-\infty}^\infty g(t) \exp(- 2\pi i t \xi) dt \\
    &= \frac{1}{\epsilon} \left\{ \hat{\varphi} \left( \frac{1}{\epsilon} \left[ \xi - \frac{\theta + \theta_0}{2\pi} \right] \right)
    + \hat{\varphi} \left( \frac{1}{\epsilon} \left[ \xi - \frac{\theta - \theta_0}{2\pi} \right] \right) \right\}~,
\end{split}
\end{equation}
therefore the lemma follows from Poisson's summation formula.
\end{proof}

Now we need some properties that are specific for $\varphi_q$ from (\ref{eq:phiq}).
Denote
\begin{equation}\label{eq:Fq}
F_q(\xi) = \int_{-\infty}^\infty \exp(-2\pi i x \xi - x^{2q}) dx~.
\end{equation}
This is obviously an entire function.
The following lemma can be proved using a saddle-point argument:
\begin{lemma}\label{l:fq}
For any $\delta > 0$ there exist $C_\delta, c_\delta > 0$ such that
\[ | F_q(\rho e^{i \phi}) | \leq C_\delta \exp \left\{ - c_\delta \rho^\frac{2q}{2q-1} \right\} \]
for $ - \frac{\pi}{4q} + \delta < \phi < \frac{\pi}{4q} - \delta$ and for
$\pi - \frac{\pi}{4q} + \delta < \phi < \pi + \frac{\pi}{4q} - \delta$.
\end{lemma}

Now take $\varphi = \varphi_q$ in Lemma~\ref{l:pois}; then $\widehat{\varphi_q} = F_q^2$. From Lemma~\ref{l:fq},
both sides of (\ref{eq:l.pois}) are analytic  functions of $\theta$, therefore
the equality can be extended to complex $\theta$. In particular, $f_{E_0, \epsilon}$ is bounded
(uniformly in $\epsilon > 0$) on an open interval containing $[-1, 1]$.

\vspace{2mm}

Finally, we state  -- for use in Section~\ref{s:emb} -- the definition and a couple of properties of
divided differences; we refer to the survey of de Boor \cite{dB} for the proofs.

For (distinct) $z_1, \cdots, z_E$ in the domain of definition of a function $f$, consider the $(E-1)$-th
divided difference $f[z_1, \cdots, z_{E}]$, defined as follows:
\[ f[z_1] = f(z_1)~, \quad f[z_1, \cdots, z_E] = \frac{f[z_1, \cdots, z_{E-1}] - f[z_2, \cdots, z_E]}{z_1 - z_E}~. \]
An equivalent definition is given by
\[ f[z_1,\cdots, z_E] = \sum_e \frac{f(z_e)}{\prod_{f \neq e}(z_e - z_f)}~.\]

\begin{lemma}\label{l:meanval}
Let $f$ be an $(E-1)$-times differentiable function in a convex domain $D \subset \mathbb{C}$. Then
for any distinct $z_1, \cdots, z_E \in D$ there exists $z^* \in \mathrm{conv}(z_1, \cdots, z_E)$
such that
\[f[z_1,\cdots,z_{E}] = \frac{f^{(E-1)}(z^*)}{(E-1)!}~.\]
\end{lemma}

\begin{lemma}\label{l:sumsimplex} For $z_1, \cdots, z_E \in \mathbb{C}$,
\[ \sum_{ (n_1, \cdots, n_E) \in \Delta_{E, n}} \prod_{e=1}^E z_e^{n_e} = m_{n-1}[z_1, \cdots, z_E] \, \prod_{e=1}^E z_e~, \]
where
\[ \Delta_{E,n} = \Big\{ n_1, \cdots, n_E \geq 1 \, \big| \, n_1 + \cdots + n_E = n \Big \}~,\]
and $m_{n-1}(z)=z^{n-1}$.
\end{lemma}

\section{Classification of paths}\label{s:diags}

Consider the collection $\bf{Paths}^0_{2n}$  of paths $u_0=0, u_1, u_2, \cdots, u_{2n-1}, u_{2n}=0$ in
$\mathbb{Z}(W)$ such that every edge appears an even number of times,  $u_j \neq u_{j+2}$ for
$j = 0,1, \cdots, n-2$, and $u_{n-1} \neq u_1$.

A {\em pairing} of a path is a paring of ${0, 1, \cdots, 2n-1}$, so that if $j$ is paired to $j'$, then
either $u_j = u_{j'}$ and $u_{j+1} = u_{j'+1}$, or $u_j = u_{j'+1}$ and $u_{j+1}=u_{j'}$. Every path in
$\bf{Paths}^0_{2n}$ has at least one pairing.

Now we divide the couples (path, pairing) into equivalence classes. The procedure is a slight elaboration
of the one from \cite[Section II.1]{FS}; it works as follows.

We look for a pair $(j, j')$ so that $j$ is paired to $j'$ and $j+1$ is paired to either $j'+1$ or $j'-1$, and unite the
$j$-th and the $j+1$-th edge into a single one (and the same for their counterparts.) Continuing this process,
we arrive at a multigraph $G = (V, E)$ with a marked vertex $v_0$ (corresponding to $0$) together with a path $p$
which passes every one of its edges twice. We call the equivalence class $D = (G= (V, E), p)$ a {\em diagram}
of order $1$.

Paths which contain edges passed more than twice admit more than one pairing, and therefore correspond
to more than one diagram. If a path passes a certain edge $4$ times, we correspond to it a diagram
of order $2$ in a similar way, and so forth. Thus, every path which corresponds to a high-order diagram
also corresponds to a sequence of diagrams of lower order, and the total number of paths
in $\bf{Paths}_{2n}$ can be computed using the inclusion--exclusion formula:
\begin{multline} \text{Paths}_{2n} \\
    = \sum_{j \geq 1} (-1)^{j+1} \sum_{\text{$D$ of order $j$}}
    \# \left\{ \text{paths of length $2n$ corresponding to $D$} \right\}~.\end{multline}
Let us compute the number of paths corresponding to a diagram $D = (G, p)$.
Fix an arbitrary ordering of the edges in $E$. To construct a path corresponding to $D$ in $\mathbb{Z}(W)$,
we first choose the positions $R_v \in \mathbb{Z}$ corresponding to $v \in V$, and the number $n_e$
of edges corresponding to every $e \in E$, so that $R_{v_0} = 0$ and $n_1 + \cdots + n_E = n$.

Then, for $e = (u, v)$ from $1$ to $E$,  we start a random walk from $R_u$ conditioned not to backtrack,
and denote by $\widetilde{\mathbb{P}}_e \left\{ R_u \overset{n_e}{\rightsquigarrow} R_v \right\}$ the probability
that it arrives at $R_v$ after $n_e$ steps without violating the non-backtracking conditions created
by the edges corresponding to $f < e$ from the previous steps at its first and last step.
The number of paths corresponding to $D$ is then equal to
\[ 2W \, (2W-1)^{n-1} \, \sum_{{\bf R}} \sum_{\bf n}
    \prod_{e=1}^E \widetilde{\mathbb{P}}_e \left\{ R_u \overset{n_e}{\rightsquigarrow} R_v \right\}~.\]

One can proceed using the expression due to Smilansky \cite{Sm} for $\widetilde{\mathbb{P}}_e$ in terms of the transition matrix
$P$ of the random walk on $\mathbb{Z}(W)$. Instead, let us denote by
${\mathbb{P}} \left\{ R_u \overset{n_e}{\rightsquigarrow} R_v \right\}$ the probability that
a random walk from $R_u$ and conditioned not to backtrack arrives at $R_v$ after $n_e$ steps.
Then
\[ \widetilde{\mathbb{P}}_e \left\{ R_u \overset{n_e}{\rightsquigarrow} R_v \right\}
    \approx {\mathbb{P} \left\{ R_u \overset{n_e}{\rightsquigarrow} R_v \right\}}~.\]
More formally, the left-hand side of the last equality can be expressed as the right-hand side
plus a sum of similar terms of the same form with different parameters. Therefore
we can essentially regard the above approximation as an identity.

The degree of the marked vertex in a diagram of order one is at least two, and the degree of every
other vertex is at least three. If these inequalities are saturated, the diagram is called simple.
For simplicity, let us focus on simple diagrams of order one (which were called diagrams in \cite{FS}.)

The genus of a diagram is defined as $\gamma(D) = E - V + 1$. A diagram of genus $\gamma$ satisfies
$E = 3\gamma-2$, $V = 2\gamma-1$; the number $D(\gamma)$ of diagrams of genus $\gamma$ satisfies
(see \cite{FS})
\[ (\gamma/C)^\gamma \leq D(\gamma) \leq (C\gamma)^\gamma~. \]
We remark that there is exactly one diagram of genus $\gamma=1$; it contains one vertex, and one
edge (which is a loop connecting the vertex to itself.)

\section{Embeddings into $\mathbb{Z}(W)$}\label{s:emb}

Fix a multigraph $G = (V, E)$ with a marked vertex $v_0 \in V$. For $g \in \mathbb{C}$
with $|g|=1$, set
\begin{equation}\label{eq:defemb1}
\begin{split}
\widetilde{\text{Emb}}(G)
    &= \widetilde{\text{Emb}}(G; g, \epsilon) \\
    &= \sum_{\bf R} \sum_{{\bf n}: E \to \mathbb{N}}
        \varphi_q(\sum_{e \in E} n_e \epsilon) \prod_{e=(u, v) \in E} g^{n_e} {\widetilde{\mathbb{P}}_e \left\{ R_u \overset{n_e}{\rightsquigarrow} R_v \right\}}~,
\end{split}
\end{equation}
where the exterior sum is over ${\bf R}: V \to \mathbb{Z}$ such that $R(v_0) = 0$.

Our goal in this section is to prove

\begin{prop}\label{p:dest}\hfill
\begin{enumerate}
\item For any $g$ with $|g|=1$,
\[ |\widetilde{\text{Emb}} (G)| \leq \left( \frac{C(q)}{|1-g|} \right)^{E+1} \, \left(\frac{\log W}{W}\right)^{E-V+1}~, \]
where $C(q)>0$ is a constant depending only on $q$.
\item If $G$ is the multigraph corresponding to the (unique) diagram of genus $\gamma=1$,
the same bound holds without the logarithmic factor.
\end{enumerate}
\end{prop}

\begin{rmk*}
The logarithmic factor is probably unnecessary in the general case as well; this is however not essential
for our purposes.
\end{rmk*}

To make the computations more transparent, we make several simplifications. First, set
\begin{equation}\label{eq:defemb}
\begin{split}
{\text{Emb}}(G)
    &= {\text{Emb}}(G; g, \epsilon) \\
    &= \sum_{\bf R} \sum_{{\bf n}: E \to \mathbb{N}}
        \varphi_q(\sum_{e \in E} n_e \epsilon) \prod_{e=(u, v) \in E} g^{n_e} {{\mathbb{P}}_e \left\{ R_u \overset{n_e}{\rightsquigarrow} R_v \right\}}~,
\end{split}
\end{equation}

By the argument sketched in the previous section, it is sufficient to prove the bound for $\text{Emb}$.

Let $P$ be the transition matrix of the usual random walk on $\mathbb{Z}(W)$. From \cite{ABLS},
\begin{equation}\label{eq:abls}
{\mathbb{P} \left\{ R_u \overset{n_e}{\rightsquigarrow} R_v \right\}}
    = \frac{1}{\sqrt{(2W-1)^{n_e}}} U_{n_e, W}\left(\frac{W P}{\sqrt{2W-1}}\right)(R_u, R_v)~,
\end{equation}
where the last brackets stand for taking matrix elements. Instead of $\text{Emb} (G)$,
we shall prove the bound for
\begin{equation}\label{eq:defemb2}
\begin{split}
{\text{Emb}^\#}(G)
    &= {\text{Emb}^\#}(G; g, \epsilon) \\
    &= \sum_{\bf R} \sum_{{\bf n}: E \to \mathbb{N}}
        \varphi_q(\sum_{e \in E} n_e \epsilon) \prod_{e=(u, v) \in E} g^{n_e} P^{n_e}(R_u, R_v)~.
\end{split}
\end{equation}
Using (\ref{eq:abls}), one can repeat the argument and obtain the same bound for $\text{Emb} (G)$
(and hence also for $\widetilde{\text{Emb}}(G)$.)

We start with a representation of $\text{Emb}^\#(G)$ in Fourier space. The operator $P$
is translation-invariant, therefore diagonal in Fourier space: setting $e_\xi(n) = \exp(2\pi i \xi n)$,
$0 \leq \xi < 1$, we have:
\begin{equation}\label{eq:pfour}
P e_\xi =  w(\xi)\,e_\xi~,
\end{equation}
where
\[ w(\xi)
    = \frac{1}{W} \sum_{j=1}^W \cos(2 \pi j\xi)
    = \frac{\sin (\pi W\xi)}{W \sin (\pi \xi)} \, \cos (\pi (W+1)\xi)~.\]
For future reference, we remark that
\begin{equation}\label{eq:wbound}
|w(\xi)| \leq \frac{1}{1 + c W \min(\xi, 1-\xi)}~.
\end{equation}

Choose an ordering of the vertices of $G$, and for every edge $e = (u, v) \in E$, $u \prec v$, introduce a variable
$\xi_e = \xi_{(u,v)}$; if $u \neq v$, set $\xi_{(v, u)} = - \xi_{(u,v)}$.
\begin{lemma}\label{l:frepr}
Set $S_\epsilon(z) = \sum_{n \geq 1} \varphi_q(n\epsilon) z^{n-1}$. Then
\[ \text{Emb}^\#(G) = \idotsint\limits_{[0, 1]^E} d\delta_\text{Kirch}({\bm \xi})
    S_\epsilon[gw(\xi_1), \cdots, gw(\xi_E)] \, \prod_{e \in E} (gw(\xi_e))~, \]
where $\delta_\text{Kirch}$ is the Lebesgue measure restricted to the $(E-V+1)$-dimensional
subspace defined by the Kirchhoff constraints
\[ \forall u \in V \quad \sum_{(u, v) \in E~, \, v \neq u} \xi_{(u, v)} = 0~. \]
\end{lemma}

\begin{proof}
From (\ref{eq:pfour}),
\[ P = \int_0^1 w(\xi) e_\xi \otimes e_\xi d\xi~, \]
and
\begin{equation}\label{eq:frepr}
P^n(R_1, R_2) = \int_0^1 w(\xi)^n \exp(2\pi i \xi \, (R_1 - R_2)) d\xi~.
\end{equation}
Now substitute (\ref{eq:frepr}) into (\ref{eq:defemb2}). We obtain:
\begin{multline*}
\text{Emb}^\#(G) \\
= \sum_{\bf R} \sum_{\bf n} \varphi(\sum_e n_e \epsilon)
    \idotsint\limits_{[0, 1]^E} \prod_{e = (u,v)} \Big[ d\xi_e (gw(\xi_e))^{n_e} \exp(2\pi i \xi_e \, (R_u - R_v)) \Big]~.
\end{multline*}
Exchanging the summation over ${\bf R}$ with the integral, we see that
\begin{equation}
\text{Emb}^\#(G) = \sum_{\bf n} \varphi(\sum_e n_e \epsilon)
    \idotsint\limits_{[0, 1]^E} d\delta_\text{Kirch}({\bm \xi}) \prod_{e = (u,v)} (gw(\xi_e))^{n_e}~.
\end{equation}
Now we can exchange the sum over ${\bf n}$ with the integral. According to Lemma~\ref{l:sumsimplex},
\[ \sum_{n_1 + \cdots + n_E = n} \prod_e (gw(\xi_e))^{n_e}
    = m_{n-1}[gw(\xi_1), \cdots, gw(\xi_E)] \, \prod_e (gw(\xi_e))~, \]
therefore
\[\begin{split}
&\sum_{\bf n} \varphi(\sum_e n_e \epsilon) \prod_{e = (u,v)} (gw(\xi_e))^{n_e} \\
    &\qquad = \sum_{n \geq 1} \varphi(n \epsilon) m_{n-1}[gw(\xi_1), \cdots, gw(\xi_E)]\,\, \prod_e (gw(\xi_e)) \\
    &\qquad = S_\epsilon[gw(\xi_1), \cdots, gw(\xi_E)]\,\, \prod_e (gw(\xi_e))~.
\end{split}\]
\end{proof}

According to the mean-value theorem (Lemma~\ref{l:meanval}),
\begin{equation}\label{eq:mv4s}
S_\epsilon[gw(\xi_1), \cdots, gw(\xi_E)] = \frac{S_\epsilon^{(E-1)}(gw^*)}{(E-1)!}
\end{equation}
for some $-1 < w^* < 1$. To conclude the proof of Proposition~\ref{p:dest}, we need one more lemma:

\begin{lemma}\label{l:s.eps} For $j \geq 0$ and $|z|\leq 1$~,
\[ |S_\epsilon^{(j)}(z)| \leq \left(\frac{C(q)}{|1-z|}\right)^{j+1} j!~,\]
where the constant $C(q)$ does not depend on $\epsilon>0$.
\end{lemma}

\begin{proof}[Proof of Proposition~\ref{p:dest}]

Applying Lemma~\ref{l:frepr}, the relation (\ref{eq:mv4s}), and Lem\-ma~\ref{l:fq},
we obtain:
\begin{equation}
|\text{Emb}^\# (G)|
    \leq \left(\frac{C(q)}{|1-g|}\right)^{E+1} \int d\delta_{\text{Kirch}}({\bm \xi}) \prod_{e \in E} |w(\xi_e)| d\xi~.
\end{equation}
The integration is over an $(E-V+1)$-dimensional subspace,  therefore the bound (\ref{eq:wbound})
concludes the proof of 1.\\

The estimate 2.\ can be verified directly from Lemma~\ref{l:frepr}.

\end{proof}

\begin{proof}[Proof of Lemma~\ref{l:s.eps}]
First,
\[
\varphi_q(n\epsilon)
    = \int_{-\infty}^\infty \exp(2\pi i x \xi) \, \hat{\varphi}_q(\xi) \, d\xi
    = \int_{-\infty}^\infty \exp(2\pi i n \epsilon \xi) \, F_q(\xi)^2 \, \frac{d\xi}{A_q}~.
\]
Fix $0<\phi<\pi/2$ satisfying the assumptions of Lemma~\ref{l:fq}, and deform the contour of integration
to $L = \left\{ \xi \, \big| \, \arg \xi \in \{\phi, \pi - \phi \} \right\}$. Now
\[\begin{split} S_\epsilon(z)
    &= \int_L F_q(\xi)^2 \sum_{n \geq 1} \exp(2\pi i n \epsilon \xi) z^{n-1}\, \frac{d\xi}{A_q} \\
    &= \int_L F_q(\xi)^2  \, \frac{\exp(2\pi i \epsilon \xi)}{(1 - z \exp(2\pi i \epsilon \xi))} \, \frac{d\xi}{A_q}~,
\end{split}\]
and
\[ S_\epsilon^{(j)} (z) = j! \, \int_L F_q(\xi)^2 \,
    \frac{\exp(2\pi i \epsilon \xi)}{(1 - z \exp(2\pi i \epsilon \xi))^{j+1}} \, \frac{d\xi}{A_q}~. \]
Taking absolute values and applying Lemma~\ref{l:fq}, we conclude the proof.

\end{proof}

\section{Proof of Theorem}\label{s:pf}

Let $q$ be a large natural number, and let $\eta>0$ be a small real number; we shall
choose them later. We shall work with the function $\varphi_q$ from (\ref{eq:phiq}),
or rather with its truncated version
\begin{equation}
\widetilde\varphi(t) = \varphi_q(t) \mathbb{1}_{\{|t| \leq W^\eta\}}.
\end{equation}
Denote
\begin{equation}
\widetilde{f}(E) = 1 + 2 \sum_{n\geq1} \widetilde\varphi_q(n\epsilon) T_n(E_0) \, T_n(E)~.
\end{equation}

\begin{lemma}\label{l:ftilde}
For any $q > 50$ and $\eta < 1/100$
\begin{equation}\label{eq:l:ftilde}
\langle \widetilde{f}(H)(0, 0) \rangle = 1 + \varphi_q(2\epsilon)(1 - 2E_0^2) + O(1/W)~.
\end{equation}
\end{lemma}

\begin{proof}

First,
\[\begin{split}
&\langle \widetilde{f}(H)(0, 0) \rangle \\
&\quad= 1 - \varphi_q(2\epsilon) (2E_0^2 - 1) \frac{W-1}{W-1/2}+ 2 \sum_{n = 2}^{n_0} \varphi_q(2n\epsilon) T_{2n}(E_0) \langle T_{2n}(H)(0, 0) \rangle~,
\end{split}\]
where $n_0 = \lfloor W^\eta / \epsilon \rfloor$.
By (\ref{eq:reprtn}),
\[\langle T_{2n}(H)(0, 0) \rangle
= \frac{1}{2(2W-1)^{n/2}} \,\, \text{Paths}^0_{2n} + \Big[\cdots\Big]~,
\]
where the last brackets enclose the terms of higher order which can be analysed similarly to the leading term,
Now apply the classification of paths described in Section~\ref{s:meth}. To simplify the notation,
we explicitly write the contribution of simple diagrams, and collect all the rest in the remainder term. This
yields:
\begin{equation}
\begin{split}  &\text{Paths}_{2n}^0 \, / \, (2W-1)^n \\
&\qquad= \sum_{\gamma \leq n_0} \sum_{\gamma(D) = \gamma} \sum_\mathbf{R} \sum_{n_1 + \cdots + n_E = n}
    \prod_{e = (u, v)}  \widetilde{P}_e\left\{ R_u \overset{n_e}{\rightsquigarrow} R_v \right\}  + \Big[\cdots\Big]~;
\end{split}
\end{equation}
therefore,
\begin{equation}\label{eq:sumn}
\begin{split}
&\sum_n \widetilde\varphi(2n\epsilon) T_{2n}(E_0) \frac{\text{Paths}^0_{2n}}{(2W-1)^n} \\
&\qquad= \Re \sum_{n\leq n_0} \widetilde\varphi(2n\epsilon) e^{2in\theta_0} \frac{\text{Paths}^0_{2n}}{(2W-1)^n} \\
&\qquad= \Re \sum_{\gamma \leq n_0} \sum_{\gamma(D) = \gamma} \sum_{n \leq n_0}
    \varphi_q(2n\epsilon) e^{2in\theta_0}\\
&\qquad\qquad\qquad\qquad\sum_\mathbf{R} \sum_{n_1 + \cdots + n_E = n} \prod_{e = (u, v)}
    \widetilde{P}_e\left\{ R_u \overset{n_e}{\rightsquigarrow} R_v \right\}  + \Big[\cdots\Big]~.
\end{split}
\end{equation}
It is not hard to see (cf.\ \cite{edge}) that
\begin{multline}\label{eq:smallterm}
\varphi_q(2n\epsilon)\sum_\mathbf{R} \sum_{n_1 + \cdots + n_E = n} \prod_{e = (u, v)}
    \widetilde{P}_e\left\{ R_u \overset{n_e}{\rightsquigarrow} R_v \right\} \\
    \leq C \exp\left[ - c n^{2q} \epsilon^{2q} \right]  \frac{(Cn)^{\frac{5\gamma-4}{2}}}{(\frac{5\gamma-4}{2})!}~;
\end{multline}
therefore the left-hand side of (\ref{eq:smallterm}) is very small for $n > n_0$ if
\begin{equation}\label{eq:restr.1}
\frac{2q-1}{2q} > 0.99~.
\end{equation}
Under this condition the sum over $n$ in (\ref{eq:sumn}) can be extended to infinity. Thus
\[\begin{split}
&\langle \widetilde{f}(H)(0, 0) \rangle \\
&\quad= 1 - \varphi_q(2\epsilon)(2E_0^2-1) + \sum_{\gamma \leq n_0} \sum_{D = (G, p), \, \gamma(D) = \gamma} \widetilde{\text{Emb}}(G) + \Big[\cdots\Big] + O(1/W)~.
\end{split}\]

For $D$ of order $\gamma$, Proposition~\ref{p:dest} and the subsequent remarks yield
\begin{equation}\label{eq:embest}
\left| \widetilde{\text{Emb}} (G) \right| \leq \left( \frac{C(q)}{|1-g|} \right)^{3\gamma} \,
    \left(\frac{\log W}{W}\right)^\gamma~,
\end{equation}
and the logarithmic factor is redundant for $\gamma=1$. As $n_0 = \lfloor W^\eta/\epsilon\rfloor \ll W$,
the estimate (\ref{eq:embest}) implies:
\begin{equation}\begin{split}
&\left| \langle \widetilde{f}(H)(0, 0) \rangle - 1 + \varphi_q(2\epsilon)(2E_0^2-1) \right| \\
&\quad\leq \frac{C}{|1-g|W}
    + \sum_{2 \leq \gamma \leq n_0} (C\gamma)^\gamma \, \left( \frac{C(q)}{|1-g|} \right)^{3\gamma} \, \frac{1}{W^\gamma}
    + O(1/W) \\
&\quad= O(1/W)~.
\end{split}\end{equation}

\end{proof}

The Chebyshev polynomials $T_n$ satisfy
\[ \max_{E \in [-W^{1/2}, W^{1/2}]} |T_n(E)| \leq (CW)^{n/2}~. \]
Therefore for $-W^{1/2} \leq E \leq W^{1/2}$
\[ \left| \sum_{n > n_0}  \varphi_q(n\epsilon) T_n(E_0) T_n(E) \right|
    \leq \sum_{n > n_0} C \exp \left[ - c n^{2q} \epsilon^{2q} + C n \ln W \right] = O(1/W) \]
as long as $2q\eta > \eta + 0.99$, or:
\begin{equation}\label{eq:restr.2}
q > \frac{\eta + 0.99}{2\eta}~.
\end{equation}
The spectrum of $H$ lies in
\[ [-\frac{W}{\sqrt{2W-1}}, \frac{W}{\sqrt{2W-1}}] \subset [-W^{1/2}, W^{1/2}]~, \]
thus, under the assumption (\ref{eq:restr.2}), the
conclusion of Lemma~\ref{l:ftilde} remains valid for $f_{E_0, \epsilon}$ in place of $\widetilde{f}$.

Next, one can replace $f_{E_0, \epsilon}$ with $f_{E_0, \epsilon}\big|_{[-1,1]}$ (using Lemma~\ref{l:fq},
the remark following it, and the fact \cite{edge} that the density of states is small outside $[-1, 1]$.)
Since
\[ h_{E_0, \epsilon} = f_{E_0, \epsilon}\big|_{[-1,1]} \Big/ (\pi \sqrt{1-E_0^2}) \geq 0 \]
is an approximate $\delta$-function of width $\epsilon$ at $E_0$, and $E_0, \epsilon$ are arbitrary
(subject to the constraint $\epsilon \geq W^{-0.99}$), we can replace it with any other approximate
$\delta$-function $\widetilde{h}_{E_0, \epsilon}$, such as the Stieltjes kernel that appears in the statement
of the theorem. Indeed, for any $\widetilde{\epsilon} = \widetilde{\epsilon}(W) \gg \epsilon(W)$,
and any $-1 < \widetilde{E}_0 < 1$, one can approximate
$\int \widetilde{h}_{\widetilde{E}_0, \widetilde{\epsilon}} \, dN$
by positive linear combinations of $\int h_{E_0, \epsilon}\,dN$ (with different $E_0$.)
\qed

\section{Two remarks}\label{s:rem}

{\bf i.} As we remarked in Section~\ref{s:meth}, the divergence of the ``na\"{\i}ve'' perturbation series
(obtained from the self-energy renormalisation procedure) follows from the divergent contribution of the spectral
edges. It is probable that a similar reason is responsible for the divergence of perturbation series also in other
problems, such as the density of states in the Anderson model (see Erd\H{o}s--Salmhofer--Yau \cite{ESY}.)

\vspace{1mm} \noindent
{\bf ii.} The restriction $\epsilon \geq W^{-0.99}$ in the main theorem appears for the following reason. It is an
artefact of the approach that (\ref{eq:1term}) can be justified for a given $\epsilon > 0$ only together with the first
$\approx 1/\epsilon$ terms of (\ref{eq:formal}). However, only the first $\approx W$ terms of (\ref{eq:formal})
are reasonably small (say, smaller than $1$), therefore we do not see how to make the current approach work for
$\epsilon \ll W^{-1}$. It is possible that the power $-0.99$ can be improved to $-1$ using a more careful
choice of the test function $\varphi$.


\begin{thebibliography}{99}

\bibitem{ABLS} N.~Alon, I.~Benjamini, E.~Lubetzky, S.~Sodin,
    Non-backtracking random walks mix faster,
    Commun.\ Contemp.\ Math., Vol.~9, No.~4 (August 2007)

\bibitem{BMP} L.~V.~Bogachev, S.~A.~Molchanov, L.~A.~Pastur,
    On the density of states of random band matrices,
    Mat.\ Zametki 50 (1991), no.~6, 31–-42, 157;
    translation in Math.\ Notes~50 (1992), no.~5–-6, 1232–-1242

\bibitem{dB} C.~de Boor,
    Divided differences,
    Surv.\ Approx.\ Theory~1 (2005), 46–--69

\bibitem{EK} L.~Erd\H{o}s, A.~Knowles,
    Quantum Diffusion and Eigenfunction Delocalization in a Random Band Matrix Model,
    Commun.\ Math.\ Phys.~303 (2011), 509–-554

\bibitem{EK2} L.~Erd\H{o}s, A.~Knowles,
    Quantum Diffusion and Delocalization for Band Matrices with General Distribution,
    arXiv:1005.1838

\bibitem{ESY} L.~Erdos, M.~Salmhofer, H.--T.~Yau,
    Feynman graphs and renormalization in quantum diffusion,
    Quantum field theory and beyond, 167–-182, World Sci.\ Publ., Hackensack, NJ, 2008.

\bibitem{ESchY} L.~Erd\H{o}s, H.~T.~Yau, J.~Yin,
    Bulk universality for generalized Wigner matrices,
    arxiv:1001.3453

\bibitem{FS} O.~N.~Feldheim, S.~Sodin,
    A universality result for the smallest eigenvalues of certain
        sample covariance matrices,
    Geom.\ Funct.\ Anal. 20-1 (2010), 88-123

\bibitem{jmp} S.~Sodin,
    Random matrices, nonbacktracking walks, and orthogonal polynomials,
    J.\ Math.\ Phys.~48 (2007), no.~12, 123503, 21 pp.

\bibitem{edge} S.~Sodin,
    The spectral edge of some random band matrices,
    Ann.\ of Math.\ 172 (2010), No.~3, 2223–-2251, arXiv:0906.4047

\bibitem{Sp} T.~Spencer,
    Random Banded and Sparse Matrices (Chapter 23),
    {\em to appear in } ``Oxford Handbook of Random Matrix Theory'',
    edited by G. Akemann, J. Baik, and P. Di Francesco,
    Oxford Handbooks in Mathematics,
    Oxford University Press, Oxford, 2011, 960 pp.

\bibitem{Sm} U.~Smilansky,
    Quantum chaos on discrete graphs,
    J.~Phys.~A 40 (2007), no.~27, F621–-F630.

\end{thebibliography}
\end{document}